\documentclass[12pt,reqno]{article}

\usepackage[english]{babel}
\usepackage{amssymb,amsmath,amsthm}
\usepackage{amsfonts}
\usepackage{euscript}

\allowdisplaybreaks[4]

\newtheorem{pro}{Proposition}
\newtheorem{cor}{Corollary}
\newtheorem{lemma}{Lemma}

\newcommand{\cprime}{\/{\mathsurround=0pt$'$}}
\DeclareMathOperator{\rank}{rank}
\let\phi=\varphi
\newcommand{\eval}[2][\right]{\relax
  \ifx#1\right\relax \left.\fi#2#1\rvert}

\hoffset=-2cm \textheight=26cm \textwidth=18.5cm \voffset=-3.7cm \date{August 
  25, 2016} \renewcommand{\abstractname}{} \title{ \protect\vspace*{-16mm} \bf
  Infinitely many nonlocal conservation laws\\ for the $ABC$ equation with
  $A+B+C\neq 0$\vspace{-3mm}}
\author{I.S. Krasil{\cprime}shchik$^{1,2}$,
  A. Sergyeyev$^3$\thanks{Corresponding author}, O.I. Morozov$^4$\\
  $^1$ Independent University of Moscow,\\ B. Vlasyevskiy per. 11, 119002
  Moscow,
  Russia\\
  $^2$ Russian State University for the Humanities\\
  Miusskaya sq.\! 6, Moscow, GSP-3, 125993, Russia\\
  $^3$ Mathematical Institute, Silesian University in Opava,\\
  Na Rybn\'\i{}\v{c}ku 1, 74601 Opava, Czech Republic\\
  $^4$ Faculty of Applied Mathematics, AGH University of Science and
  Technology,\\ Al. Mickiewicza 30, 30059 Krak\'ow, Poland\\
  E-mails: \texttt{josephkra@gmail.com}, {\tt Artur.Sergyeyev@math.slu.cz},\\
  {\tt morozov@agh.edu.pl}}
\begin{document}\maketitle
\renewcommand{\abstractname}{}
\begin{abstract}\vspace*{-16mm}
  We construct an infinite hierarchy of nonlocal conservation laws for the
  $ABC$ equation $A u_t\,u_{xy}+B u_x\,u_{ty}+C u_y\,u_{tx} = 0$, where
  $A,B,C$ are nonzero constants and $A+B+C\neq 0$, using a nonisospectral Lax
  pair.  As a byproduct, we present new coverings for the equation in
  question. The method of proof of nontriviality of the conservation laws
  under study is quite general and can be applied to many other integrable
  multidimensional systems.

{\bf Keywords:} integrable systems; conservation laws; Lax pairs.

{\bf MSC 2010:} 37K05; 37K10.\vspace{-3mm}
\end{abstract}

\section{Introduction}
Integrable systems are well known to play an important role in modern
mathematics, both pure and applied, see e.g.\ \cite{ac, b, ca, cb, c, c2, c3, d,
  g, hel1, hel2, h, o, ps} and references therein. Existence of an infinite
hierarchy of conservation laws is among the most important features of
integrable\footnote{In the present paper we mean by integrability existence of
  a nontrivial Lax pair for the system under study, cf.\ e.g. \cite{ac, d,
    v-z} and references therein for details.}  systems of partial differential
equations \cite{ac, b,d,o}. It imposes 
strong constraints on the associated dynamics making it highly
regular.\looseness=-1

While such a hierarchy of conservation laws can often be extracted from the
Lax-type representation of the system under study, cf.\ e.g.\ \cite{b,d} and
references therein, in a relatively straightforward manner, rigorous proof of
nontriviality and independence of the conservation laws in question is often a
tricky matter, especially in the case of integrable systems of partial
differential equations in more than two independent variables when the
conservation laws under study often happen to be nonlocal, see e.g.\ \cite{b, kv}.\looseness=-1

In the present paper we demonstrate how to prove nontriviality and
independence of such nonlocal conservation laws at the example of the $ABC$
equation.  The procedure presented below is based on the careful examination
of the structure of the kernel of total derivatives and is fairly readily
generalized to other multidimensional integrable systems.

Recall that the $ABC$ equation has the form
\begin{equation}
A u_t\,u_{xy}+B u_x\,u_{ty}+C u_y\,u_{tx} = 0,
\label{abc}
\end{equation}
where $A,B,C$ are arbitrary nonzero constants (if one of them vanishes,
\eqref{abc} reduces to a first-order PDE).\looseness=-1

To the best of our knowledge, equation~\eqref{abc} has first appeared in
\cite{z} in connection with the study of geometry of Veronese webs (cf.\ also \cite{kp} and references therein). In the same
paper the author has established integrability of~\eqref{abc} for the case of
$A+B+C=0$ by presenting the associated Lax pair; note that in \cite{mos}
a four-dimensional integrable generalization of~\eqref{abc} with
$A+B+C=0$ was found. Later in \cite{ms} (cf.\ also
\cite{ms2, om2, s, s2} for related results) a recursion operator for~\eqref{abc}
with $A+B+C=0$ was found, and using the method of hydrodynamic reductions it
was shown \cite{f} that~\eqref{abc} is also integrable if $A+B+C\neq 0$. Note
also that if $A=B=C\neq 0$ then~\eqref{abc} admits \cite{fkt} a Lagrangian
with the density $-A u_x u_y u_t/2$.\looseness=-1

Below we assume that $A+B+C \neq 0$ and put $B = -\kappa_1 A$,
$C = -\kappa_2 A$, so $\kappa_1+\kappa_2-1\neq 0$. Then equation~\eqref{abc}
takes the form
\begin{equation}
  u_{xy} = \frac{\kappa_1 \, u_x\,u_{ty}+\kappa_2\,u_y\,u_{tx}}{u_t}.
\label{abc_2}
\end{equation}

Integrability of~\eqref{abc_2} is an immediate consequence of the following
result which provides a nonlinear Lax pair for the equation in question.
\begin{pro}[\cite{f}]\label{p1}
The $ABC$ equation~\eqref{abc_2} has a covering defined by the system
\begin{equation}
\begin{array}{lcl}
  q_y &=& \displaystyle{u_y\,
          \frac{\kappa_1\,s\,R^{\prime}+
          (\kappa_1+\kappa_2-1)\,R}{\kappa_1\,R^{\prime}},}\\[5mm]
q_t &=& \displaystyle{u_t\,\frac{\kappa_1\,s\,R^{\prime}+
          R^2+(\kappa_1+\kappa_2-1)\,R}{\kappa_1\,R^{\prime}},}
\end{array}
\label{abc_2_covering}
\end{equation}
where $s = q_x/u_x$\textup{,} and the function $R=R(s)$ is a solution to the
ODE
\begin{equation}
  R^{\prime\prime} = \frac{(2\,(\kappa_1+\kappa_2)-1)\,
    R+(\kappa_1+\kappa_2-1)\,(2\,\kappa_1+\kappa_2-1)}
  {(\kappa_1+\kappa_2-1)\,R\,(R+\kappa_1+\kappa_2-1)}\,(R^{\prime})^2.
  \label{abc_2_ode}
\end{equation}
\end{pro}
To make the exposition self-contained, we give a very brief introduction to the
theory of differential coverings in Section~\ref{sec:diff-cover}; for more
details of this theory we refer the reader to \cite{b, k-v0, k-v1, kv} and
references therein; the discussion of how certain types of nonlinear coverings
are related to integrability
can be found e.g.\ in \cite{d,om1,v-z}.\looseness=-1

\noindent {\bf Remark 1.}  Equation~\eqref{abc_2_ode} is integrable by
quadratures. Its general solution is of the form $R(s)=\Omega(c_1\,s+c_2)$,
where $\Omega$ is the inverse function for the function
\[
\omega(z) = \int z^{-(2\kappa_1+\kappa_2-1)/(\kappa_1+\kappa_2-1)}
\,(z+\kappa_1+\kappa_2-1)^{-\kappa_2/(\kappa_1+\kappa_2-1)} d z
\]
and $c_1$, $c_2$ are arbitrary constants.

Most importantly, (\ref{abc_2_covering}) gives rise\footnote{See e.g.\
  \cite[Subsection 10.3.3]{d} and \cite{fk} and references therein for the
  general construction leading from a covering of the type
  (\ref{abc_2_covering}) to a Lax pair of the type (\ref{abc_noniso_lax-0}),
  and \cite{d,man} and references therein for nonisospectral Lax pairs in
  general.} to a linear nonisospectral Lax pair
for~\eqref{abc_2}:\looseness=-1
\begin{cor}\label{int} 
  The $ABC$ equation~\eqref{abc_2}
  admits a linear nonisospectral Lax pair of the form
\begin{equation}\label{abc_noniso_lax-0}
\begin{array}{rcl}
\Phi_y&=&(H_1)_{\zeta}\Phi_x-(H_1)_x\Phi_\zeta,\\[3mm]
\Phi_t&=&(H_2)_{\zeta}\Phi_x-(H_2)_x\Phi_\zeta,
\end{array}
\end{equation}
where
\[
H_1=\left(\displaystyle{u_y\, \frac{\kappa_1\,s\,R^{\prime}+
      (\kappa_1+\kappa_2-1)\,R}{\kappa_1\,R^{\prime}}}\right)_{s=\zeta/u_x},\quad
H_2=\left(\displaystyle{u_t\,\frac{\kappa_1\,s\,R^{\prime}+
      R^2+(\kappa_1+\kappa_2-1)\,R}{\kappa_1\,R^{\prime}}}\right)_{s=\zeta/u_x}.
\]
\end{cor}
The coefficients of~\eqref{abc_noniso_lax-0} depend on the variable spectral
parameter $\zeta$ in a highly nontrivial fashion thanks to the presence of the
function $R$, so our first order of business is to construct a Lax
representation with simpler dependence on the parameter. This is done in
Section~\ref{nc}, where we present a new covering for~\eqref{abc_2} and the
associated nonisospectral Lax pair~\eqref{abc_2_noniso} and show how it is
related to (\ref{abc_noniso_lax-0}).  Finally, in Section~\ref{cl} we
construct an infinite hierarchy of nonlocal conservation laws
for~\eqref{abc_2} and prove their nontriviality.  Note that the method of the
proof is quite general and can be applied to many other integrable
multidimensional systems.

\section{Differential coverings}
\label{sec:diff-cover}

We briefly review here the theory of differential coverings over infinitely
prolonged differential equations. The reader can find further details and examples in~\cite{b, k-v0, k-v1, kv}.

Let~$M$ be a smooth manifold, $\dim M=n$, and~$\pi\colon E\to M$ be a vector
bundle, $\rank\pi=m$. We consider the bundles of
$k$-jets~$\pi_k\colon J^k(\pi) \to M$, $k\geq0$, together with the natural
projections $\pi_{k+1,k}\colon J^{k+1}(\pi)\to J^k(\pi)$. Then the manifold of
infinite jets~$J^\infty(\pi)$ is defined as the inverse limit with respect to
these projections and the bundles $\pi_\infty\colon J^\infty(\pi)\to M$ and
$\pi_{\infty,k}\colon J^\infty(\pi)\to J^k(\pi)$ are defined as well. For any
section~$s\colon M\to E$ of the bundle~$\pi$ its infinite
jet~$j_\infty(s)\colon M\to J^\infty(\pi)$ is a section of~$\pi_\infty$. One
has the embeddings
$\pi_{k+1,k}^*\colon C^\infty(J^k(\pi))\to C^\infty(J^{k+1}(\pi))$, and we
define the algebra of smooth functions on~$J^\infty(\pi)$ as
$\mathcal{F}(\pi)=\cup_{k\geq0}C^\infty(J^k(\pi))$.

The main geometric structure on~$J^\infty(\pi)$ is the Cartan
distribution~$\mathcal{C}$: for any point~$\theta\in J^\infty(\pi)$ we define
the Cartan plane~$\mathcal{C}_\theta$ as the tangent plane to the graph of an
infinite jet passing through this point. This distribution is formally
integrable: if~$X$ and~$Y$ are vector fields lying in~$\mathcal{C}$ then the
commutator~$[X,Y]$ lies there as well. Every Cartan plane~$\mathcal{C}_\theta$
is $n$-dimensional and projects isomorphically to~$T_{\pi_\infty(\theta)}M$ by
the differential of~$\pi_\infty$. Due to this, any vector field~$Z$ on~$M$ can
be uniquely lifted to a vector field~$\mathcal{C}_X$ on~$J^\infty(\pi)$. The
correspondence~$X\mapsto \mathcal{C}_X$ is $C^\infty(M)$-linear and preserves
the commutator. In addition, $\pi_{\infty,*}\left(\mathcal{C}_X\right) = X$.
In other words, we have a connection which is called the Cartan connection. In the
standard local coordinates $x^1,\dots,x^n$, $u_\sigma^1,\dots,u_\sigma^m$
in~$J^\infty(\pi)$, $\sigma$ being symmetric multi-index consisting of the
integers~$1,\dots,n$, the Cartan connection is determined by the
correspondence
\begin{equation}\label{eq:5}
  \mathcal{C}\colon\frac{\partial}{\partial x^i}\mapsto D_{x^i} =
  \frac{\partial}{\partial x^i} + \sum_{\sigma,j}u_{\sigma
    i}^j\frac{\partial}{\partial u_\sigma^j},
\end{equation}
where the fields~$D_{x^i}$ are called the total derivatives. Differential
operators in total derivatives are called $\mathcal{C}$-differential
operators.

A differential equation\footnote{By a slight abuse of terminology, we speak of a differential equation even though it could actually be a system of differential equations.} of order~$k$ is a submanifold in~$J^k(\pi)$. Locally,
it can be given by the conditions~$F^1=\dots=F^r=0$, where~$F^j$ are smooth
functions on~$J^k(\pi)$. Its infinite prolongation is the submanifold
(probably, with singularities)~$\mathcal{E}$ in~$J^\infty(\pi)$ satisfying
the conditions~$(D_{x^{i_1}}\circ\dots\circ D_{x^{i_s}})(F^j)=0$ for
all~$j=1,\dots,r$, $s\geq 0$, and~$1\leq i_1,\dots,i_s\leq n$. The Cartan
connection can be restricted from the bundle~$\pi_\infty$ to its
subbundle~$\eval{\pi_\infty}_{\mathcal{E}}\colon\mathcal{E}\to M$ and
hence any $\mathcal{C}$-differential operator restricts
from~$J^\infty(\pi)$ to~$M$. On the other hand, the
correspondence~\eqref{eq:5} allows one to lift any linear differential
operator on~$M$ to a $\mathcal{C}$-differential operator on~$\mathcal{E}$. We
always assume that~$\mathcal{E}$ is differentially connected which means that
the only solutions of the system~$D_{x^i}(f)=0$, $i=1,\dots,n$,
on~$\mathcal{E}$ are constants.

In particular, let~$\ell_{\mathcal{E}}$ denote the restriction of the
linearization operator
\begin{equation*}
  \ell_F =
  \begin{pmatrix}
    \sum_\sigma\dfrac{\partial F^\alpha}{\partial u_\sigma^\beta}D_\sigma
  \end{pmatrix},\qquad
  \alpha=1,\dots,r,\quad \beta=1,\dots,m,
\end{equation*}
to~$\mathcal{E}$. Then the solutions of the
equation~$\ell_{\mathcal{E}}(\phi) =0$ are identified with symmetries
of~$\mathcal{E}$, i.e., with the evolutionary vector fields
$\sum_{j,\sigma} D_\sigma(\phi^j)\partial/\partial u_\sigma^j$ that preserve
the Cartan distribution on~$\mathcal{E}$. Dually, the solutions
of~$\ell_{\mathcal{E}}^*(\psi)=0$ are called cosymmetries,
where~$\ell_{\mathcal{E}}^*$ is formally adjoint to~$\ell_{\mathcal{E}}$. The
lift~$d_h$ of the de~Rham differential gives rise to the horizontal de~Rham complex
on~$\mathcal{E}$; $d_h$-closed $(n-1)$-forms are conservation laws
of~$\mathcal{E}$ and $d_h$-exact forms are trivial conservation laws. To any
conservation law~$\omega$ one can associate its generating
section (or the characteristic)~$\psi_\omega$ which is a cosymmetry.

A morphism of equations is a smooth
map~$\tau\colon\tilde{\mathcal{E}} \to \mathcal{E}$ which takes the Cartan
distribution~$\tilde{\mathcal{C}}$ on~$\tilde{\mathcal{E}}$ to that
on~$\mathcal{E}$. A morphism~$\tau$ is a (differential) covering if its
differential maps the Cartan plane~$\mathcal{C}_{\tilde{\theta}}$
to~$\mathcal{C}_{\tau(\tilde{\theta})}$ isomorphically for
any~$\tilde{\theta}\in\tilde{\mathcal{E}}$. In other words, for any vector
field~$Z$ on~$M$ the field~$\tilde{\mathcal{C}}_Z$ projects
to~$\mathcal{C}_Z$. Locally, this means that the total derivatives
on~$\tilde{\mathcal{E}}$ are
\begin{equation}\label{eq:6}
  \tilde{D}_{x^i} = D_{x^i} + X_i,\qquad i=1,\dots,n,
\end{equation}
where~$D_{x^i}$ are the total derivatives on~$\mathcal{E}$, and
\begin{equation*}
  D_{x^i}(X_j)-D_{x^j}(X_i) + [X_i,X_j]=0,\qquad 1\leq i<j\leq n,
\end{equation*}
$X_i$ being $\tau$-vertical vector fields on~$\tilde{\mathcal{E}}$. A
covering~$\tau\colon \tilde{\mathcal{E}} \to \mathcal{E}$ over a
differentially connected equation is called irreducible if the covering
equation~$\tilde{\mathcal{E}}$ is differentially connected as well.

Symmetries, cosymmetries, conservation laws of the covering
equation~$\tilde{\mathcal{E}}$ are nonlocal symmetries, etc.,
of~$\mathcal{E}$. Local objects depend on formal solutions of~$\mathcal{E}$ and
their partial derivatives; roughly speaking, nonlocal ones may depend on integrals of these
solutions\footnote{Or, more generally, on solutions of some differential
  equations whose coefficients depend on formal solutions
  of~$\mathcal{E}$.}.

For example, the relations
\begin{equation}\label{cov-burgers}
  \tilde{D}_x=D_x+\frac{wu}{2},\qquad
  \tilde{D}_t=D_t+\frac{w}{2}\left(\frac{u^2}{2} + u_x\right)
\end{equation}
define a covering of the Burgers equation~$\mathcal{E}=\{u_t=uu_x+u_{xx}\}$ by
the heat equation~$\tilde{\mathcal{E}} = \{w_t=w_{xx}\}$.

Thus, $w$ is related to $u$ by the formulas
\begin{equation}\label{w}
\displaystyle w_x=\frac{wu}{2},\quad w_t=\frac{w}{2}\left(\frac{u^2}{2} + u_x\right).
\end{equation}
Note that system (\ref{w}) is compatible by virtue of the Burgers equation.

The
form~$\omega=w\,dx+w_x\,dt$ is a local conservation law
of~$\tilde{\mathcal{E}}$, and its pullback to $\mathcal{E}$ gives a nonlocal conservation law for
$\mathcal{E}$. The corresponding nonlocal conserved density on~$\mathcal{E}$, i.e., $w$, defined by (\ref{w}),
can be informally thought of as $\int\exp(\frac{u}{2})\,dx$.\looseness=-1

Going back to the general theory, note that any $\mathcal{C}$-differential
operator~$\Delta$ on~$\mathcal{E}$ can be lifted to a
$\mathcal{C}$-differential operator~$\tilde{\Delta}$ on~$\tilde{\mathcal{E}}$
using equations~\eqref{eq:6}. In particular, this can be done with the
linearization operator~$\ell_{\mathcal{E}}$ and its adjoint. Solutions of the
equations
\begin{equation*}
  \tilde{\ell}_{\mathcal{E}}(\phi)=0,\qquad \tilde{\ell}_{\mathcal{E}}^*(\psi)=0
\end{equation*}
are called nonlocal shadows of symmetries and cosymmetries, respectively.

We employ the theory of coverings to establish nontriviality and independence of
nonlocal conservation laws constructed in Section~\ref{cl}.

\section{New coverings and nonisospectral Lax pair}\label{nc}
We can readily write down a new covering for~\eqref{abc_2} expressed solely in
terms of the variable $r=R(s)$:

\begin{pro}\label{p2} The $ABC$ equation~\eqref{abc_2} has a covering defined by
  the system
\begin{equation}
    \begin{array}{lcl}
      r_t &=& \displaystyle{\frac{r((\kappa_1 + \kappa_2 - 1)(r + \kappa_1 +
              \kappa_2 - 1) u_{tx} - u_t r_x)}{u_x \kappa_1(\kappa_1 + \kappa_2 -
              1)},}
      \\
          &&
      \\
      r_y &=& \displaystyle{\frac{r((\kappa_1 + \kappa_2 - 1)(r + \kappa_1 +
              \kappa_2 - 1) u_{xy} - \kappa_2 u_y r_x)}{u_x \kappa_1(r +
              \kappa_1 +
              \kappa_2 - 1)}}.
    \end{array}
  \label{abc_2_r_covering}
\end{equation}
\end{pro}

Recall that a symmetry shadow (resp.\ a cosymmetry) for~\eqref{abc_2} is a
solution of linearized (resp.\ adjoint linearized) version of~\eqref{abc_2},
see Section~\ref{sec:diff-cover} and~\cite{b, k-v0, k-v1, kv} for further
details. 


\begin{pro}\label{p3} The $ABC$ equation~\eqref{abc_2} has a shadow of nonlocal
  symmetry in the covering~\eqref{abc_2_r_covering} with the characteristic
  \begin{equation}
    U=\int \displaystyle
    (r+\kappa_1+\kappa_2-1)^{\frac{(1-\kappa_1)}{(\kappa_1+\kappa_2-1)}}
    r^{\frac{(1-\kappa_2)}{(\kappa_1+\kappa_2-1)}} dr.
    \label{abc_2_r_covering_shadow}
  \end{equation}
\end{pro}

\begin{pro}\label{p4} The $ABC$ equation~\eqref{abc_2} has a nonlocal cosymmetry
  with the characteristic
  \begin{equation}
    \gamma=u_t \int \displaystyle
    (r+\kappa_1+\kappa_2-1)^{\frac{-2\kappa_2}{(\kappa_1+\kappa_2-1)}}
    r^{\frac{-2\kappa_1}{(\kappa_1+\kappa_2-1)}} dr.
    \label{abc_2_r_covering_cosymm}
  \end{equation}
\end{pro}

It can be shown that the shadow~\eqref{abc_2_r_covering_shadow} cannot be
lifted to a nonlocal symmetry for the $ABC$ equation~\eqref{abc_2} in the
covering~\eqref{abc_2_r_covering}.\looseness=-1

Now pass from~\eqref{abc_2_r_covering} to a slightly different (but
equivalent) covering by putting $w=r
u_x^{-(\kappa_1+\kappa_2-1)/\kappa_1}$. Then
we have
\begin{equation}
    \begin{array}{lcl}
      w_t &=& \displaystyle{-\frac{u_x^{(\kappa_2-\kappa_1-1)/\kappa_1}w
              (\kappa_1 u_x u_t w_x+(\kappa_1+\kappa_2-1)w (u_t
              u_{xx}-\kappa_1 u_x u_{xt}))}{\kappa_1^2(\kappa_1+\kappa_2-1)},}
      \\
          &&
      \\
      w_y &=& \displaystyle{-\frac{\kappa_2
              u_x^{(\kappa_2-\kappa_1-1)/\kappa_1} w u_y(\kappa_1 u_x w_x+
              (\kappa_1+\kappa_2-1)w u_{xx})}{\kappa_1^2
              (u_x^{(\kappa_1+\kappa_2-1)/\kappa_1}w+\kappa_1+\kappa_2-1)}}.
    \end{array}
  \label{abc_2_w_covering}
\end{equation}

Moreover, a similar change of variables in (\ref{abc_noniso_lax-0})
leads to a significantly simpler Lax pair for (\ref{abc_2}):
\begin{pro}\label{lax}
  The $ABC$ equation~\eqref{abc_2} admits a nonisospectral Lax pair of the form
  \begin{equation}
      \begin{array}{lcl}
        \Psi_y &=& \displaystyle{\frac{\kappa_2 \lambda
                   u_x^{(\kappa_2-\kappa_1-1)/\kappa_1}  u_y(-\kappa_1 u_x
                   \Psi_x+ (\kappa_1+\kappa_2-1)\lambda
                   u_{xx}\Psi_\lambda)}{\kappa_1^2
                   (u_x^{(\kappa_1+\kappa_2-1)/\kappa_1}\lambda+
                   \kappa_1+\kappa_2-1)},}
        \\
               &&
        \\
        \Psi_t &=& \displaystyle{\frac{\lambda
                   u_x^{(\kappa_2-\kappa_1-1)/\kappa_1} (-\kappa_1 u_x u_t
                   \Psi_x+(\kappa_1+\kappa_2-1)\lambda (u_t u_{xx}-\kappa_1
                   u_x
                   u_{xt})\Psi_\lambda)}{\kappa_1^2(\kappa_1+\kappa_2-1)},}
      \end{array}
    \label{abc_2_noniso}
  \end{equation}
  which is related to~\eqref{abc_noniso_lax-0} by the transformation
  $\Phi=\Psi$, $\lambda=u_x^{-(\kappa_1+\kappa_2-1)/\kappa_1}R(\zeta/u_x)$.
\end{pro}
Note that \eqref{abc_2_noniso} can be obtained
from~\eqref{abc_2_w_covering} via the so-called Pavlov eversion, see
\cite[Section~2]{p} and \cite{fk, k} for details on the latter.

In stark contrast with (\ref{abc_noniso_lax-0}), the variable spectral parameter $\lambda$ enters
(\ref{abc_2_noniso}) rationally. This enables us to construct an infinite sequence of conservation laws
from (\ref{abc_2_noniso}) using a formal Taylor expansion for $\Psi$ with respect to $\lambda$ in the fashion outlined below.

\section{Nonlocal conservation laws}\label{cl}
Substituting into~\eqref{abc_2_noniso} a formal expansion
$\Psi=\sum\limits_{j=-\infty}^\infty \psi_j\lambda^j$ yields the 
equations
\begin{equation}
  \hspace*{-5.7mm}
    \begin{array}{lcl}
      (\psi_j)_y &=& \displaystyle{\frac{\kappa_2
                     u_x^{(\kappa_2-\kappa_1-1)/\kappa_1}  u_y(-\kappa_1 u_x
                     (\psi_{j-1})_x+ (\kappa_1+\kappa_2-1)(j-1)
                     u_{xx}\psi_{j-1})}{\kappa_1^2
                     (\kappa_1+\kappa_2-1)}
                     -\frac{u_x^{(\kappa_1+\kappa_2-1)/\kappa_1}(\psi_{j-1})_y}{(\kappa_1+\kappa_2-1)}},
      \\[5mm]
      (\psi_j)_t &=&\displaystyle{\frac{u_x^{(\kappa_2-\kappa_1-1)/\kappa_1}
                     (-\kappa_1 u_x u_t (\psi
                     _{j-1})_x+(\kappa_1+\kappa_2-1)(j-1) (u_t u_{xx}-\kappa_1
                     u_x u_{xt})\psi_{j-1})}{\kappa_1^2(\kappa_1+\kappa_2-1)}}
    \end{array}
  \label{abc_2_psi}
\end{equation}
for $j\in\mathbb{Z}$.

However, the covering over (\ref{abc_2}) defined by (\ref{abc_2_psi}) for $j\in\mathbb{Z}$ is pretty much intractable, and there appears
to be no way to extract from it any reasonably simple conservation laws for (\ref{abc_2}).

Fortunately, the situation improves dramatically when we truncate the expansion for $\Psi$. One natural possibility to do this is to pass from the Laurent expansion for $\Psi$ to the Taylor one, i.e., to
assume that $\psi_j=0$ for $j<0$. The substitution of $\Psi=\sum\limits_{j=0}^\infty \psi_j\lambda^j$ into (\ref{abc_2_noniso})
yields (\ref{abc_2_psi}) for $j=1,2,3,\dots$, and the equations
\begin{equation}
  (\psi_0)_t=0,\quad (\psi_0)_y=0.\label{psi0}
\end{equation}

System~\eqref{abc_2_psi} for $j=1,2,3,\dots$ together with
equations~\eqref{psi0} defines an infinite-dimensional covering
over~\eqref{abc_2} and yields, cf.\ e.g.\ \cite{bkmv, bkmv2, ks, as} and
references therein, an infinite sequence of (in general nonlocal)
two-component conservation laws for~\eqref{abc_2} of the form
\begin{equation}
  \hspace*{-16mm}
  \begin{array}{lcl}
    \tilde{D}_y\left(\displaystyle{\frac{u_x^{(\kappa_2-\kappa_1-1)/\kappa_1}
    (-\kappa_1 u_x u_t (\psi_{j-1})_x+(\kappa_1+\kappa_2-1)(j-1) (u_t
    u_{xx}-\kappa_1 u_x
    u_{xt})\psi_{j-1})}{\kappa_1^2(\kappa_1+\kappa_2-1)}}\right)=\\[7mm]
    \tilde{D}_t\left(\displaystyle{\frac{\kappa_2
    u_x^{(\kappa_2-\kappa_1-1)/\kappa_1}  u_y(-\kappa_1 u_x (\psi_{j-1})_x+
    (\kappa_1+\kappa_2-1)(j-1) u_{xx}\psi_{j-1})}{\kappa_1^2
    (\kappa_1+\kappa_2-1)}-
    \frac{u_x^{(\kappa_1+\kappa_2-1)/\kappa_1}(\psi_{j-1})_y}{(\kappa_1+\kappa_2-1)}}\right)
  \end{array}
  \label{abc_2_cl}
\end{equation}
for $j=1,2,3\dots$. The operators $\tilde{D}_y$ and $\tilde{D}_t$ denote here
the total $y$- and $t$-derivatives in the covering~\eqref{abc_2_psi}, cf.\
e.g.\ \cite{b, kv} and the discussion after Lemma~\ref{l2} below for details.
The nonlocal variables and the associated conservation laws are constructed
recursively.

The simplest choice $\psi_0=0$, and even a more general choice
$\psi_0=\mathrm{const}$, yields by virtue of~\eqref{abc_2_psi}
\[
(\psi_1)_t=0,\quad (\psi_1)_y=0.
\]
If we now choose $\psi_1=1$, so $\Psi=\lambda+\sum\limits_{j=2}^\infty\psi_j\lambda^j$, then we have
\begin{equation}
  \hspace*{-25mm}
    \begin{array}{lcl}
      (\psi_2)_y &=& \displaystyle{\frac{\kappa_2
                     u_x^{(\kappa_2-\kappa_1-1)/\kappa_1}  u_y
                     u_{xx}}{\kappa_1^2}},
      \\[7mm]
      (\psi_2)_t&=& \displaystyle{\frac{u_x^{(\kappa_2-\kappa_1-1)/\kappa_1}
                    (u_t u_{xx}-\kappa_1 u_x u_{xt})}{\kappa_1^2}},
    \end{array}
  \label{abc_2_psi2}
\end{equation}
which gives rise to the first nontrivial conservation law in the sequence (\ref{abc_2_cl}),
\begin{equation}\label{abc_eq_lcl}
  D_y\left(\displaystyle{\frac{u_x^{(\kappa_2-\kappa_1-1)/\kappa_1} (u_t
        u_{xx}-\kappa_1 u_x
        u_{xt})}{\kappa_1^2}}\right)=D_t\left(\displaystyle{\frac{\kappa_2
        u_x^{(\kappa_2-\kappa_1-1)/\kappa_1}  u_y
        u_{xx}}{\kappa_1^2}}\right).
\end{equation}
This conservation law is local but the conservation laws~\eqref{abc_2_cl} for $j=3,4,\dots$
will be
nonlocal:

\begin{pro}\label{p5} The $ABC$ equation~\eqref{abc_2} has an infinite sequence
  of nontrivial nonlocal linearly independent conservation
  laws~\eqref{abc_2_cl} for $j=3,4,\dots$\textup{,} where $\psi_0=0$\textup{,}
  $\psi_1=1$\textup{,} the nonlocal variable $\psi_2$ is defined
  by~\eqref{abc_2_psi2}\textup{,} and the nonlocal variables
  $\psi_j$\textup{,} $j=3,4,\dots$\textup{,} are defined recursively
  via~\eqref{abc_2_psi}.
\end{pro}

\begin{proof}


  We begin with a general construction. Let~$\mathcal{E}$ be a differentially
  connected equation and suppose that
  $\omega = A_1\,dx^1\wedge\,dx^3\wedge\dots\wedge\,dx^n +
  A_2\,dx^2\wedge\,dx^3\wedge\dots\wedge\,dx^n$
  is a nontrivial two-component conservation law on $\tilde{\mathcal{E}}$,
  i.e.,
  \begin{equation*}
    D_{x^2}(A_1)=D_{x^1}(A_2),
  \end{equation*}
  where $D_{x^i}$ are the total derivatives on~$\mathcal{E}$. Consider the
  covering $\tau_\omega\colon \mathcal{E}_\omega \to \mathcal{E}$ naturally
  associated to~$\omega$. This covering contains nonlocal
  variables~$\psi_\sigma$, where $\sigma$ is a multi-index whose components
  take values in the set $\{3,\dots,n\}$, that satisfy the defining equations
  \begin{equation}\label{eq:7}
    \begin{array}{rcl}
    (\psi_\sigma)_{x^1}&=&D_\sigma(A_1),\\
    (\psi_\sigma)_{x^2}&=&D_\sigma(A_2),\\
    (\psi_\sigma)_{x^i}&=&\psi_{\sigma i},
    \end{array}
  \end{equation}
  where $i=3,\dots,n$.

  Thus, infinitely many two-component conservation laws
  of the from
  \begin{equation*}
    \omega_\sigma=D_\sigma(A_1)\,dx^1\wedge\,dx^3\wedge\dots\wedge\,dx^n +
    D_\sigma(A_2)\,dx^2\wedge\,dx^3\wedge\dots\wedge\,dx^n
  \end{equation*}
  arise on $\mathcal{E}$.

  A straightforward generalization of the results proved in~\cite{IntInCov} is
  the following

  \begin{lemma}\label{sec:nonl-cons-laws-2}
    If the equation~$\mathcal{E}$ is differentially connected then the
    conservation laws $\omega_\sigma$ are linearly independent if and only if
    the only solutions of the system
    \begin{equation*}
      \tilde{D}_{x^1}(f)=\tilde{D}_{x^2}(f)=0
    \end{equation*}
    are functions of $x^3,\dots,x^n$\textup{,} where~$\tilde{D}_{x^i}$ are the
    total derivatives on~$\mathcal{E}_\omega$.
  \end{lemma}

  From equations~\eqref{eq:7} it also immediately follows that under the
  assumptions of Lemma~\ref{sec:nonl-cons-laws-2} the covering~$\tau_\omega$
  is irreducible.

  Now return to the $ABC$ equation~\eqref{abc}. Using the
  presentation~\eqref{abc_2} introduce the internal variables
  \begin{equation}\label{i-v}
    u_k=\frac{\partial^ku}{\partial t^k},\qquad
    v_{k,l}=\frac{\partial^{k+l}u}{\partial t^k\partial x^l},\qquad
    w_{k,l}=\frac{\partial^{k+l}u}{\partial t^k\partial y^l},
  \end{equation}
  where $k\geq 0$, $l\geq 1$. The total derivatives in these coordinates take
  the form
  \begin{align*}
    D_t&=\frac{\partial}{\partial t} +\sum_k u_{k+1}\frac{\partial}{\partial
         u_k} + \sum_{k,l}\left(v_{k+1,l}\frac{\partial}{\partial v_{k,l}} +
         w_{k+1,l}\frac{\partial}{\partial w_{k,l}}\right),\\
    D_x&=\frac{\partial}{\partial x} +\sum_k v_{k,1}\frac{\partial}{\partial
         u_k} +
         \sum_{k,l}\left(v_{k,l+1}\frac{\partial}{\partial v_{k,l}} +
         D_t^kD_y^{l-1}(\Upsilon)\frac{\partial}{\partial w_{k,l}}\right),\\
    D_y&=\frac{\partial}{\partial y} +\sum_kw_{k,1}\frac{\partial}{\partial u_k}
         + \sum_{k,l}\left(D_t^kD_x^{l-1}(\Upsilon)\frac{\partial}{\partial
         v_{k,l}} + w_{k,l+1}\frac{\partial}{\partial w_{k,l}}\right),
  \end{align*}
  where $\Upsilon$ denotes the right-hand side of~\eqref{abc_2}.

  From the above formulas we immediately obtain
  \begin{lemma}\label{l2}
    The $ABC$ equation~\eqref{abc_2} is differentially connected.
  \end{lemma}
 Now denote by $Y_k$ the right-hand sides of first equations in
  systems~\eqref{abc_2_psi} and~\eqref{abc_2_psi2} and by $T_k$ the right-hand
  sides of second equations in these systems, $k\geq 2$. Then we have
  \begin{equation*}
    \tilde{D}_y=
    D_y+\sum_{j\geq2}\sum_{k\geq0}\tilde{D}_x^k(Y_j)\frac{\partial}{\partial
      \psi_{j,k}},\quad
        \tilde{D}_t=
    D_t+\sum_{j\geq2}\sum_{k\geq0}\tilde{D}_x^k(T_j)\frac{\partial}{\partial
      \psi_{j,k}},
  \end{equation*}
  where $\psi_{j,0}=\psi_j$ and $\psi_{j,k+1}=(\psi_{j,k})_x$.

  The subsequent lemma is proved by direct computations using the
  expressions~\eqref{abc_2_psi2} and~\eqref{abc_2_psi} for $T_j$:
  \begin{lemma}\label{sec:nonl-cons-laws}
    The following estimates hold\textup{:}
    \begin{equation}
      \label{eq:1}
      \tilde{D}_x^k(T_2)=\frac{v_{0,1}^\alpha}{\kappa_1^2}
      (u_1v_{0,k+2}-\kappa_1v_{0,1}v_{1,k+1}) +o(2,k)
    \end{equation}
    and
    \begin{equation}
      \label{eq:2}
      \tilde{D}_x^k(T_j)=(j-1)\frac{v_{0,1}^\alpha}{\kappa_1^2}
      (u_1v_{0,k+2}-\kappa_1v_{0,1}v_{1,k+1})\psi_{j-1,0}
      -\frac{v_{0,1}^{\alpha+1}u_1}{\kappa_1(\kappa_1
        +\kappa_2-1)}\psi_{j-1,k+1} + o(j,k),\quad j>2,
    \end{equation}
    where 
    $\alpha=(\kappa_2-\kappa_1-1)/\kappa_1$\textup{,} and $o(j,k)$ is a
    function that does not depend on the variables $\psi_{s,l}$ for $s>j-1$,
    $\psi_{j-1,l}$ for $l>k+1$\textup{,} and $v_{0,l+2}$\textup{,} $v_{1,l+1}$
    for~$l>k$.
  \end{lemma}

 We are now ready to establish a stronger result:
  \begin{lemma}
    \label{sec:nonl-cons-laws-1}
    The only solutions of the equation
    \begin{equation}\label{eq:3}
      \tilde{D}_t(f)=0
    \end{equation}
    are functions of~$x$ and~$y$.
  \end{lemma}
  \begin{proof}[Proof of Lemma~\ref{sec:nonl-cons-laws-1}]
    Suppose that
    \begin{equation}\label{eq:4}
      f=f(x,y,t,\dots,\psi_{2,0},\dots,\psi_{2,k_2},\dots,\psi_{j,0},\dots,
      \psi_{j,k_j})
    \end{equation}
    is a solution of equation~\eqref{eq:3}. Let us stress that in addition to $x,y,t$ and $\psi_{j,k}$ the
    function $f$ is allowed to depend on finitely many internal
    variables~\eqref{i-v} on the $ABC$ equation.

    Now proceed by induction on~$j$.

    \textbf{Base of induction:} $j=2$. Let
    $f=f(x,y,t,\dots,\psi_{2,0},\dots,\psi_{2,k_2})$. Then from~\eqref{eq:1}
    one has
    \begin{equation*}
      \tilde{D}_t(f) +
      \frac{v_{0,1}^\alpha}{\kappa_1^2}
      \sum_{k=0}^{k_2}\big((u_1v_{0,k+2}-\kappa_1v_{0,1}v_{1,k+1})
      +o(2,k)\big)\frac{\partial f}{\partial
        \psi_{2,k}} = 0.
    \end{equation*}
    We now perform induction on $k_2$ and show that $f$ cannot depend on
    nonlocal variables $\psi_{2,0},\dots,\psi_{2,k_2}$. Indeed, for $k_2=0$
    one has
    \begin{equation*}
      \tilde{D}_t(f) +
      \frac{v_{0,1}^\alpha}{\kappa_1^2}
      \big((u_1v_{0,2}-\kappa_1v_{0,1}v_{1,1})
      +o(2,0)\big)\frac{\partial f}{\partial
        \psi_{2,0}} = 0.
    \end{equation*}
    But from the structure of the operator $\tilde{D}_t$ it immediately follows that
    $f$ is independent of~$v_{0,2}$, and
    thus~$\partial f/\partial\psi_{2,0}=0$; hence the claim holds true. If
    $k_2>0$ then for similar reasons $f$ cannot depend on $v_{0,k_2+2}$ and
    consequently $f=f(x,y,t,\dots,\psi_{2,0},\dots,\psi_{2,k_2-1})$.
    However, this form of $f$ contradicts our initial assumption that
    $\partial f/\partial \psi_{2,k_2}\neq 0$, and hence $f$ is local.

    \textbf{Induction step:} $j>2$. Let $f$ be of the form~\eqref{eq:4}. Then
    from the estimate~\eqref{eq:2} it follows that~$k_{j-1}>k_j$; otherwise
    $f$ will be independent of~$\psi_{j,k_j}$. Repeating this argument, we obtain
    \begin{equation*}
      k_2>k_3>\dots>k_{j-1}>k_j.
    \end{equation*}
    Using now the estimate~\eqref{eq:1}, we see that the coefficient
    at~$\partial/\partial\psi_{2,k_2}$ linearly depends on~$v_{0,k_2+2}$,
    where~$k_2+2\geq k_j+j$. But this dependence is impossible for the reasons
    similar to those used in the proof of the case~$j=2$. Consequently, $f$ is
    independent of all~$\psi_{j,k}$, $k=0,\dots,k_j$, and we arrive at the
    situation of the induction hypothesis, which completes the proof of
    Lemma~\ref{sec:nonl-cons-laws-1}.
  \end{proof}
  From Lemma~\ref{sec:nonl-cons-laws-1} it now follows that the only solutions
  of the equation~$\tilde{D}_t(f) = \tilde{D}_y(f) = 0$ are functions of~$x$ alone,
  and the result of Proposition~\ref{p5} now readily follows from
  Lemma~\ref{sec:nonl-cons-laws-2}.
\end{proof}

\section{Closing remarks}

In the present paper we have constructed an infinite hierarchy of nonlocal
conservation laws (\ref{abc_2_cl}) for the $ABC$ equation (\ref{abc_2}) and
proved the nontriviality of those. An interesting byproduct of our work is the change of spectral
parameter which simplifies the original nonisospectral Lax pair
(\ref{abc_noniso_lax-0}) down to (\ref{abc_2_noniso}).\looseness=-1 

Let us reiterate that the method of proving nontriviality of the nonlocal
conservation laws in question is quite general and can be applied {\em mutatis mutandis}
to many other multidimensional integrable systems.

In addition to being an important integrability attribute {\em per se}, the
conservation laws found in our paper can be employed, by means of the associated potentials $\psi_j$, $j=2,3,\dots$, 
for the
construction of nonlocal symmetries, nonlocal cosymmetries and further nonlocal conservation laws 
for the equation under study, cf., for example, \cite{b, cb, k-v1, kv} and references therein. 

On the other hand, just 
as the local conservation laws, the nonlocal ones give rise, once we perform a suitable change of independent 
variables and then rewrite our equation as an evolutionary system, 
to nonlocal integrals of motion, cf.\ e.g.\  \cite{cb, o}. On the more speculative side, 
perhaps it could have been possible to apply the nonlocal conservation laws for the construction 
of exact solutions of the equation under study using a suitably adapted version of the
method of conservation laws \cite{i}.\looseness=-1 

In closing note that in addition to the hierarchy of nonlocal conservation
laws (\ref{abc_2_cl}), all of which are two-component, the $ABC$ equation
(\ref{abc_2}) also admits {\em inter alia} a three-component conservation
law of the form \cite{f}
\begin{equation}\label{cl3}
D_x\left((\kappa_1-\kappa_2+1)u_y u_t\right) + D_y\left((1-\kappa_1+\kappa_2) u_x u_t\right) -
D_t\left((\kappa_1+\kappa_2+1) u_x u_y\right)= 0.
\end{equation}
This shows, in particular, that the set of conservation laws (\ref{abc_eq_lcl}), (\ref{abc_2_cl}) for (\ref{abc_2})
is by no means complete.

In fact, acting on the conservation law (\ref{cl3}) by appropriately chosen (nonlocal) symmetries, or even by shadows,
is likely to yield many more (nonlocal) three-component conservation laws for (\ref{abc_2}).
More broadly, finding non-two-component nonlocal conservation laws for
multidimensional integrable systems is an interesting problem on its own right which is,
however, beyond the scope of the present paper (cf.\ also the discussion in \cite{cb} and references therein
for the systems which are not necessarily integrable).\looseness=-1

\subsection*{Acknowledgments}
The research of AS was supported in part by the Grant Agency of the Czech
Republic (GA \v{C}R) under grant P201/12/G028 and by the Ministry of Education,
Youth and Sports of the Czech Republic (M\v{S}MT \v{CR}) under RVO funding for
I\v{C}47813059. The work of ISK was partially supported by the Simons-IUM
fellowship. OIM gratefully acknowledges financial support from the Polish Ministry of Science and Higher Education. 
AS is pleased to thank E.V.\ Ferapontov and R.O. Popovych for stimulating
discussions.\looseness=-1

This research was initiated in the course of visits of OIM to Silesian
University in Opava and of AS to the AGH University of Science and Technology. The
authors thank the universities in question for warm hospitality extended to
them.

The authors thank the editor and the anonymous referee for useful suggestions.

\end{document}